\documentclass[journal]{IEEEtran}

\usepackage[pdftex]{graphicx}

\usepackage{algorithm}

\usepackage[noend]{algpseudocode}

\usepackage{url}

\usepackage[dvipsnames]{xcolor}

\newcommand{\RISCV}{\mbox{RISC-V}}

\usepackage{pifont}
\newcommand{\cmark}{\ding{51}}%
\newcommand{\xmark}{\ding{55}}%
\definecolor{ColorCmark}{RGB}{0,144,0}

\usepackage{amsthm}
\newtheorem{theorem}{Theorem}
\theoremstyle{definition}
\newtheorem{definition}{Definition}

\usepackage[nolist]{acronym}

\usepackage{amsmath}
\usepackage{amsfonts}
\usepackage{mathabx} 

\usepackage{hyperref}
\usepackage{tikz}

\newcommand\copyrighttext{\scriptsize %
	DOI: \href{https://doi.org/10.1109/TCAD.2024.3374249}{10.1109/TCAD.2024.3374249}. %
	\quad\textcopyright~2024 IEEE. Personal use of this material is permitted.
	Permission from IEEE must be obtained for all other uses, in any current or future
	media, including reprinting/republishing this material for advertising or promotional
	purposes, creating new collective works, for resale or redistribution to servers or
	lists, or reuse of any copyrighted component of this work in other works.
}

\newcommand\copyrightnotice{%
	\begin{tikzpicture}[remember picture,overlay]
		\node[anchor=south,yshift=7pt] at (current page.south) {\fbox{\parbox{\dimexpr\textwidth-\fboxsep-\fboxrule\relax}{\copyrighttext}}};
	\end{tikzpicture}%
}

\begin{document}
	
	\begin{acronym}
		\acro{hw}[HW]{hardware}
		\acro{sw}[SW]{software}
		\acro{tee}[TEE]{Trusted Execution Environment}
		\acro{rtl}[RTL]{Register-Transfer Level}
		\acro{boom}[BOOM]{Berkeley Out-of-Order Machine}
		\acro{ift}[IFT]{Information Flow Tracking}
		\acro{fsm}[FSM]{finite state machine}
		\acro{ipc}[IPC]{Interval Property Checking}
		\acro{upec}[UPEC]{Unique Program Execution Checking}
		\acro{sva}[SVA]{SystemVerilog Assertions}
		\acro{isa}[ISA]{Instruction Set Architecture}
		\acro{fu}[FU]{functional unit}
		\acro{fp}[FP]{floating-point}
		\acro{rob}[ROB]{re-order buffer}
		\acro{lsu}[LSU]{load-store unit}
		\acro{csr}[CSR]{control and status register}
	\end{acronym}

	\newcommand{\TitleFirst}{%
		A Scalable Formal Verification Methodology\\%
		for Data-Oblivious Hardware%
	}
	
	\title{%
		\TitleFirst%
		\\%
	}
	
	\author{Lucas Deutschmann, %
		Johannes Müller, %
		Mohammad R. Fadiheh,\\%
		Dominik Stoffel, \IEEEmembership{Member, IEEE}, and %
		Wolfgang Kunz, \IEEEmembership{Fellow, IEEE}
		\thanks{L. Deutschmann, J. Müller, D. Stoffel, and W. Kunz are with the Department of Electrical and Computer Engineering, RPTU Kaiserslautern-Landau, Kaiserslautern, Germany.}
		\thanks{M. R. Fadiheh is with the Department of Electrical Engineering, Stanford University, Stanford, USA.}
		\thanks{The reported research was supported in part by DFG SPP Nano Security, KU 1051/11-1, in part by BMBF ZuSE (Scale4Edge), 16ME0122K-16ME0124, and in part by Intel Corporation (Scalable Assurance).}
	}

	\maketitle
	
	\copyrightnotice
	
\begin{abstract}
  The importance of preventing microarchitectural timing side channels
  in security-critical applications has surged in recent years. %
  Constant-time programming has emerged as a best-practice technique
  for preventing the leakage of secret information through timing. %
  It is based on the assumption that the timing of certain basic
  machine instructions is independent of their respective input
  data. %
  However, whether or not an instruction satisfies this data-independent
  timing criterion varies between individual processor microarchitectures. %
  
  In this paper, we propose a novel methodology to formally verify
  data-oblivious behavior in hardware using standard property checking
  techniques. %
  The proposed methodology is based on an inductive property that enables scalability even to complex out-of-order cores. %
   We show that proving this inductive property is sufficient to
  exhaustively verify data-obliviousness at the microarchitectural
  level. %
  In addition, the paper discusses several techniques that can be used to make the
  verification process easier and faster. %
  
  We demonstrate the feasibility of the proposed methodology through case studies on several open-source designs. %
  One case study uncovered a data-dependent timing violation in the extensively verified and highly secure IBEX RISC-V core. %
  In addition to several hardware accelerators and in-order processors, our experiments also include RISC-V BOOM, a complex out-of-order processor, highlighting the scalability of the approach. %

\end{abstract}


\IEEEpeerreviewmaketitle

\section{Introduction}
\label{sec:introduction}

\IEEEPARstart{I}{n} recent years, the view on \acf{hw} as a root-of-trust
has been severely damaged. %
A flood of new security vulnerabilities renewed the focus on
microarchitectural side channels. %
Both \ac{sw} and \ac{hw} communities have proposed numerous
countermeasures to these new security gaps. %
However, to fully meet the stringent demands of security-critical
applications, a holistic combination of multiple countermeasures is
needed that takes consideration of the entire system stack. %

The most prominent \ac{sw} paradigm that tries to mitigate
microarchitectural timing side channels is known as
\textit{data-oblivious}, or \textit{constant-time},
\textit{programming}~\cite{2009-CoppensVerbauwhede.etal,
  2020-CauligiDisselkoen.etal, web-bearssl, web-ctcrypto}. %
It works by the assumption that the timing and the resource usage of
certain operations inside a processor are independent of their
respective input data. %
Constant-time programming is an actively studied discipline with many
important contributions from the \ac{sw} community, including
open-source libraries~\cite{2015-AndryscoKohlbrenner.etal,
  2021-MeierColombo.etal, web-bearssl}, domain-specific
languages~\cite{2017-CauligiSoeller.etal}, verification
tools~\cite{2016-AlmeidaBarbosa.etal, 2016-RodriguesQuintao.etal,
  2017-ReparazBalasch.etal, 2020-CauligiDisselkoen.etal,
  2022-ChakrabortyCruz.etal} and dedicated
compilers~\cite{2018-AndryscoNotzli.etal, 2020-BartheBlazy.etal}. %
The term "constant-time" can, however, be misleading, as there is no
need for execution times to be \textit{constant}. %
A variable operation timing is acceptable, as long as it depends only
on public information. %
For example, in constant-time programming, the program itself is
considered public information. %
With this assumption in mind, consider a Read-After-Write (RAW) hazard
in a pipelined processor causing a stall. %
The resulting change in instruction timing is legal because it is
based on the \textit{public} sequence of instructions, not on its
(potentially confidential) operands. %

The data-oblivious subset of a processor's instruction set is often
referred to as the set of \textit{oblivious HW primitives}. %
More complex instructions, which could potentially be insecure, are
decomposed into these primitives to ensure a data-oblivious
behavior. %
As a simple example, consider a processor which implements a
multi-cycle multiplication. %
A common optimization in a HW multiplier is to check whether %
one of the operands is zero and, if so, produce the result after a single
cycle. %
This creates a data-dependent timing and a possible side channel if
the operand contains confidential information. %
Instead of issuing a multiplication, constant-time programming would
therefore resort to replacing the multiply instruction by a sequence
of primitive instructions like \textit{add} and \textit{shift}. %
However, there is no guarantee that even these simple instructions are
actually \textit{oblivious} HW primitives. %
Wether or not an instruction fulfills the criterion of
data-independent timing can vary between the implemented
microarchitectures. %
Yet, there is only little research on how to verify these assumptions
at the microarchitectural level~\cite{2017b-ArdeshirichamHu.etal,
  2019-GleissenthallKici.etal, 2021-KiciGleissenthall.etal,
  2022-DeutschmannMueller.etal}. %

To make things worse, a recent survey~\cite{2021-VicarteShome.etal}
highlighted seven classes of microarchitectural optimizations that all
undermine the constant-time paradigm. %
While exploiting some of these optimizations in commodity processors
may seem unrealistic, an attack called
Augury~\cite{2022-VicarteFlanders.etal} demonstrated that this threat
is not only a theoretical one. %
The work shows the security implication of one of these optimizations,
namely data memory-dependent prefetchers, as they are present in
modern Apple processors. %
In light of more and more such advanced
optimizations being implemented, the question arises as to how we can
restore the trust in \ac{hw}. %

To this end, we propose a novel
methodology for proving data-oblivious
execution of HW operations using standard formal property checking
techniques. %
For processors, the approach certifies a set of trusted \ac{hw}
primitives for data-oblivious programming. %
Our results show that even simple instructions like a logical shift
might suffer from an unexpected timing variability. %
We also found a potential, and preventable, timing vulnerability in
the extensively verified Ibex \RISCV{} core~\cite{web-ibex}. %
Furthermore, we extend the approach to out-of-order cores,
demonstrating its feasibility and scalability with an experiment on
the \ac{boom}~\cite{2020-ZhaoKorpan.etal}. %

In summary, this paper makes the following contributions: %
\begin{itemize}
  \item Sec.~\ref{sec:related-work} -- We provide a comprehensive
  overview of related approaches that aim to ensure
  data-obliviousness building upon an analysis at the HW level. %
  
  \item Sec.~\ref{sec:theoretical-foundation} -- We introduce a
  definition for data-oblivious execution at the microarchitectural
  level and formalize it using the notion of a 2-safety
  property over infinite traces. %
  Since most \ac{hw} designs are not built for data-obliviousness, we
  also introduce a weaker notion of \emph{input-constrained} data-obliviousness
 for general designs. %
 In order to make these properties verifiable in practice, we present how the definitions over infinite traces can be transformed to inductive properties that span over only a single clock cycle. %
 
\item Sec.~\ref{sec:methodology} -- We propose a formal verification
  methodology, called \emph{Unique Program Execution Checking for
    Data-Independent Timing (UPEC-DIT)}~\cite{2022-DeutschmannMueller.etal}, that can
  exhaustively detect any violation of data-obliviousness at the
  \ac{hw} level. %
  The methodology is based on standard formal property languages, and
  can therefore be easily integrated into existing formal verification
  flows. %
  When applied to processor implementations, UPEC-DIT can be used to qualify the instructions of a microarchitectural ISA implementation regarding their data-obliviousness. %
  These data-oblivious instructions constitute HW primitives for countermeasures such as constant-time programming and can therefore serve as a \ac{hw} root-of-trust for the higher levels of the system stack. %
  
      \item Sec.~\ref{sec:optimizations} -- We present several
        optimization techniques, such as black-boxing, proofs over an
        unrolled model, and cone-of-influence reduction, which can
        further increase the scalability and usability of the proposed
        methodology. %
  
      \item Sec.~\ref{sec:experiments} -- We demonstrate the
        feasibility of our approach through case studies on multiple
        open-source \ac{rtl} designs. %
  Besides several \ac{hw} accelerators and in-order processors, our experiments also
        cover the \ac{boom}~\cite{2020-ZhaoKorpan.etal}, a superscalar
        \RISCV{} processor with \ac{fp} support, a deep 10-stage
        pipeline and out-of-order execution. %
\end{itemize}

\section{Related Work}
\label{sec:related-work}

One of the first works that aims to formally verify data-obliviousness on the microarchitectural level is Clepsydra~\cite{2017b-ArdeshirichamHu.etal}. %
In their approach, the authors instrument the \ac{hw} with \ac{ift}~\cite{2022-HuArdeshiricham.etal, 2014-SubramanyanArora} in such a way that it tracks
not only explicit but also implicit timing flows. %
The verification engineer is then able to perform simulation, emulation or formal verification to verify timing flow properties on
the instrumented design. %
We believe that the ability to utilize a simulation-based approach can be particularly useful when dealing with very large designs. %
The experiments conducted using formal verification in Clepsydra are, however, restricted to individual functional units, e.g., cryptographic
cores. %
The additional logic added to monitor timing flows introduces a complexity overhead within the design that can be prohibitively large when trying to formally verify data-independent timing in commercial-sized processors. %

IODINE~\cite{2019-GleissenthallKici.etal} and XENON~\cite{2021-KiciGleissenthall.etal} are tools to formally verify data-obliviousness at the \ac{rtl}. %
To this end, they build their own tool chain and annotation system. %
While this approach is an important contribution to restoring trust in HW, its dependency on a custom tool environment may be an obstacle to adoption in industry. %
In contrast, in our proposed methodology, we use standard \ac{sva} which allows us to leverage any commercial (or open source) property checking tool. %
Our goal is to complement existing formal workflows, creating a synergy between functional and security verification, and thus lowering the barrier to industrial adoption. %
	
Although XENON makes considerable scalability improvements over its predecessor IODINE, it may face complexity hurdles when dealing with commercial microarchitectures. %
XENON models the entire propagation path from the data inputs (sources) to the control outputs (sinks). %
For complex systems, however, the complexity of unrolling only a few clock cycles can already be prohibitive. %
In contrast, our work employs an inductive reasoning, which results in a property that spans only a single clock cycle. %
This significantly improves the scalability of formal security verification for data-obliviousness. %
For the first time, we present a method that is applicable to large processors featuring out-of-order execution. %

A related line of research aims to establish a formally defined relationship between \ac{hw} and \ac{sw} by formulating and verifying so-called \ac{hw}-\ac{sw}~contracts~\cite{2023-WangMohr.etal}. %
Similar to~\cite{2022-DeutschmannMueller.etal}, they can be used to prove data-independent timing on an instruction-level granularity. %
However, current experiments only cover in-order processor designs up to a pipeline depth of three stages. %
A promising and ongoing related work named \textit{TransmitSynth}~\cite{2022-HsiaoMulligan.etal, 2021-HsiaoMulligan.etal} maps such a contract to a verifiable \ac{sva} property in order to automatically detect data-dependent side effects. %
 \textit{TransmitSynth} can enumerate microarchitectural execution paths for a given instruction under verification. %
This allows for a fine-grained categorization of leakage scenarios. %
However, the workflow includes a manual annotation of so-called \textit{Performing Locations}~(PLs), which are identifiers that mark an instruction execution path. %
Correctly marking these PLs requires some knowledge of the underlying design and could be increasingly difficult with more complex systems, especially for deep out-of-order processors like \ac{boom}. %
 
Other work pursues augmenting the \ac{isa} with information about the data-obliviousness of
instructions. %
In fact, both Intel~\cite{web-intel-dit} and ARM~\cite{web-arm-dit} have
  recently added support for instructions with data-independent timing. With the same goal in mind, \RISCV{} has just ratified the \textit{Cryptography Extension for Scalar \& Entropy
  Source Instructions}~\cite{web-riscvCrypto}. %
A subset of this extension, denoted \textit{Zkt}, requires a processor
to implement certain instructions of standard extensions with a
data-independent execution latency. %
This \textit{ISA contract} provides the programmer with a safe subset
of instructions that can be used for constant-time programming. %
Another work on a \textit{Data-Oblivious ISA Extension
  (OISA)}~\cite{2019-YuHsiung.etal} proposes to
refine the \ac{isa} with information about the data-obliviousness of
each instruction. %
The authors then develop \ac{hw} support for this to track wether
confidential information can reach unsafe operands. %
The method proposed here is complementary to the research efforts on such architectures as it provides a tool to verify their security. %

Finally, we point out that research on more efficient formal solvers, such as \textit{IC3/PDR}~\cite{2011-Bradley}, is complementary to the proposed approach. %
Our work proposes a solver-agnostic formal verification methodology, and any progress on more efficient model checkers can contribute to improved scalability of our method. %

\section{Theoretical Foundation}
\label{sec:theoretical-foundation}

In this section, we introduce a formal notation that we will use
throughout this paper (Sec.~\ref{subsec:definitions}). %
We formally define data-obliviousness at the microarchitectural level
(Sec.~\ref{subsec:data-obliviousness-in-hw}) and then develop a weaker
definition that is suitable for general circuits not
specifically designed for data-obliviousness
(Sec.~\ref{subsec:constrained-data-obliviousness}). %
In order to ensure scalability for more complex designs, we translate
these definitions, which are formulated over infinite traces,
into~1-cycle inductive properties
(Sec.~\ref{subsec:do-in-practice}). %
We prove that these inductive properties are equivalent to the
corresponding definitions over infinite traces. %
To conclude this chapter, we address some interesting special cases
(Sec.~\ref{subsec:special-cases}). %

\subsection{Definitions}
\label{subsec:definitions}

We first introduce some general notations to reason about
data-obliviousness as a \ac{hw} property. %
We model a digital \ac{hw} design as a
standard \ac{fsm} of Mealy type,
$M = (I, S, S_{0}, O, \delta, \lambda)$, with finite sets of input
symbols~$I$, output symbols~$O$, states~$S$, initial
states~$S_0 \subseteq S$, transition
function~$\delta: S \times I \mapsto S$ and output
function~$\lambda: S \times I \mapsto O$. %
The interface sets~$I$, $O$ % 
and the state set~$S$ are encoded in (binary-valued) input
variables~$X$, output variables~$Y$ and state variables~$Z$,
respectively. %

A key observation is that, in a \ac{hw} design, the timing of a module
is dictated by its control behavior. %
Accordingly, we partition each interface set into two disjoint subsets in order
to separate \textit{control}~($C$) from \textit{data} ($D$). %
We denote these sets as~$X_C$, $X_D$, $Y_C$ and~$Y_D$, with

\[X_C \cup X_D = X; X_C \cap X_D = \emptyset\]
\[Y_C \cup Y_D = Y; Y_C \cap Y_D = \emptyset\]

In practice, this partitioning of the interface is straightforward and
is usually done manually based on the specification of the design. %
For example, the operands and the result of a functional unit are
considered \emph{data}, whereas any handshaking signals that trigger
the start of a new computation or indicate that a provided result is
valid belong to \emph{control}. %
  
We further define a \textit{trace} $\tau=\{e_0, e_1, ...\}$ to be  a sequence of events, with an event~$e_t$ being a tuple
$(i_t,s_t,o_t)$, where~$i_t$ is the valuation of our design's input
variables~$X$ at time point (clock cycle)~$t$, $s_t$ is its state, as
represented by the value of~$Z$ at~$t$ and~$o_t$ is the valuation of
its output variables~$Y$ at~$t$. %
Let~$T$ be the set of all infinite traces of the design where $s_0 \in S_0$. %
 
We introduce the following definitions: %

\begin{itemize}
\item $s(\tau)=\{s_0, s_1, ...\}$ is the sequence of valuations to the
  design's state variables~$Z$ in the trace $\tau \in T$. %
   
\item $i(\tau)=\{i_0, i_1, ...\}$ is the sequence of valuations to the
  design's input variables~$X$ in the trace $\tau \in T$. %
  Likewise, $i_C(\tau)$ is the sequence of valuations to~$X_C$ and
  $i_D(\tau)$ is the sequence of valuations to~$X_D$. %
  
\item $o(\tau)=\{o_0, o_1, ...\}$ is the sequence of valuations to the
  design's output variables~$Y$ in the trace $\tau \in T$. %
  Likewise, $o_C(\tau)$ is the sequence of valuations to~$Y_C$ and
  $o_D(\tau)$ is the sequence of valuations to~$Y_D$. %
  
\item With a slight abuse of notation, we also allow the above functions to
    take single events as arguments, e.g., $s(e_t) = s_t$. 
  
\item For any $t \in \mathbb{N}_{0}$, the notation $\tau[t]$
  represents the event~$e_t$ at time point~$t$. %
  For example, $s(\tau[t])=s_t$ represents the valuation to the
  state variables in~$Z$ at time point~$t$. %
  
\item Similarly, we define the notion $[t..]$ as an infinite time
  interval beginning with and including~$t$. %
  Accordingly, $\tau[t..]$ represents an infinite subsequence of~$\tau$ that
  starts at and includes the time point~$t$. %
\end{itemize}

\subsection{Data-Obliviousness in HW}
\label{subsec:data-obliviousness-in-hw}

Having established a basic notation, we proceed to defining
data-oblivious behavior at the microarchitectural level. %
In our threat model, we assume that an attacker cannot access confidential information directly. %
The attacker is, however, able to observe the control signals of the
design under attack, e.g., by monitoring bus transactions. %
This means that, for a \ac{hw} module to be data-oblivious, the data
used in the computations must not affect the operating time or cause
other microarchitectural side effects. %
Based on this intuition we formulate a definition of the security
feature under verification: %

\begin{definition}[\textit{Data-Obliviousness}]
\label{def:data-oblivious}
\hfill \\
A \ac{hw} module is called \textit{Data-Oblivious (DO)} if the sequence of values of its control outputs~$Y_C$ is uniquely determined by the
sequence of values at the control inputs~$X_C$. %
\qed%
\end{definition}

We can express Def.~\ref{def:data-oblivious} formally as a 2-safety
property over infinite traces. %
We formalize it as follows: %
For any two infinite traces running on the design whose starting
  states at time~$t$ are identical and which receive the same control
  input sequences from time point~$t$ on, data-obliviousness guarantees that
  the control outputs after~$t$ are also identical: %

\begin{equation}
\label{eq:data-oblivious}
\begin{split}
  \textit{\textbf{DO}} \coloneq \
  & \forall \tau_1, \tau_2 \in T,\ \forall t \in \mathbb{N}_0: \\
  & s(\tau_1[t]) = s(\tau_2[t]) \land i_C(\tau_1[t..]) = i_C(\tau_2[t..]) \\
  & \Rightarrow o_C(\tau_1[t..]) = o_C(\tau_2[t..])
\end{split}
\end{equation}

\subsection{Input-Constrained Data-Obliviousness}
\label{subsec:constrained-data-obliviousness}

Def.~\ref{def:data-oblivious} of data-obliviousness is fairly
straightforward. %
Put simply, it ensures that the control behavior of a \ac{hw} design
is independent of the data it processes. %
Unfortunately, this strict definition works only for \ac{hw} that is carefully
designed for data-obliviousness. %
In general, however, designs are not data-oblivious. %
A processor must be able to make decisions based on the data it is
processing, for example, when it executes conditional \textit{branch} instructions. %
Constant-time programming tries to prevent
data-dependent timing by excluding such instructions from the
security-critical parts of the program. %
Data-obliviousness is achieved by restricting the program to only
use a data-oblivious subset of the ISA. %
Consequently, in order to qualify a microarchitectural implementation for
data-obliviousness, we require a separation between data-oblivious and
non-data-oblivious operations at the \ac{hw} level. %
This means, we must systematically identify and formally verify the
\textit{control input configurations} under which the design operates data-independently. %
In practice, this requires constraining  
the possible input values to a
legal subset that ensures data-obliviousness. %

\begin{definition}[Input-constrained Execution]
  \label{def:input-constraint-alt}
  \hfill \\
  An \emph{input constraint} is a non-empty subset $\phi \subseteq I$ of the possible inputs to the design. %
  An \emph{input-constrained trace} $\tau_{\phi}$ is an infinite
  trace in which the
  inputs to the design are constrained by $\phi$, i.e.,
  $i(\tau_{\phi}[t]) \in \phi$ for every $t \in \mathbb{N}$. %
  The subset $T_{\phi}\subseteq T$ of all traces constrained by
  $\phi$ is called \emph{input-constrained execution} of the
  design. %
  \qed%
\end{definition}

We can now modify Def.~\ref{def:data-oblivious} to introduce a weaker
notion of data-obliviousness. %
We call a \ac{hw} design that provides a data-oblivious subset of its functionality
\textit{input-constrained data-oblivious}. %

\begin{definition}[Input-constrained Data-Obliviousness]
  \label{def:constrained-data-oblivious}
  \hfill \\
  A \ac{hw} design is called \emph{input-constrained data-oblivious ($DO_\phi$)}
  if, for a given input constraint $\phi \subseteq I$,
  the values of its control outputs~$Y_C$ are uniquely determined by
  the sequence of control inputs~$X_C$. %
  \qed%
\end{definition}

In essence, Def.~\ref{def:constrained-data-oblivious} partitions the
design behavior into data-oblivious and non-oblivious \ac{hw}
operations. %
The \ac{hw} runs without any observable side effect on the
architectural level, as long as only inputs within the constraint
$\phi$ are given. %
For the special case of $\phi = I$,
Def.~\ref{def:constrained-data-oblivious} is equivalent to the
original Def.~\ref{def:data-oblivious}. %

The following trace property formalizes the requirement of input-constrained data-obliviousness: % 
 
\begin{equation}
  \label{eq:constrained-data-oblivious}
  \begin{split}
    \textit{\textbf{DO}}_\phi \coloneq \
    & \exists\phi, \ \forall \tau_1, \tau_2 \in T_{\phi}, \ \forall t \in \mathbb{N}_{0}: \\
    & s(\tau_1[t]) = s(\tau_2[t]) \land i_C(\tau_1[t..]) = i_C(\tau_2[t..]) \\
    & \Rightarrow o_C(\tau_1[t..]) = o_C(\tau_2[t..])
  \end{split}
\end{equation}

\subsection{Formally Proving Data-Obliviousness in Practice}
\label{subsec:do-in-practice}

 In the previous subsections, we defined data-obliviousness as a
property over infinite traces. %
 Most commercial model checking tools, however, reason about
sequential circuits by unrolling the combinatorial part of the design
over a finite number of clock cycles. %
Therefore, our definitions of data-obliviousness formalized over an
\emph{infinite} number of clock cycles are not yet suitable for being checked on practical tools. %

For exhaustive coverage of every possible design behavior, the circuit
must be unrolled to its \emph{sequential depth}. %
The sequential depth of a circuit is the minimum number of clock cycles needed to reach all possible states, typically
starting from the reset state of the design. %
For most practical designs, the sequential depth can easily
reach thousands of clock cycles. %
To make things worse, the sheer size of an industrial design may
make it impossible for the property checker to unroll for more than a few clock cycles, even for highly optimized
commercial tools. %
Therefore, it is usually not possible to verify such a design
\emph{exhaustively} with bounded model checking from the reset
state. %

\ac{ipc}~\cite{2008-NguyenThalmaier.etal, 2014-UrdahlStoffel.etal}
provides unbounded proofs and can scale to large designs  by starting from a symbolic initial
state. %
Instead of traversing a large number of states from reset, IPC starts from an arbitrary ``any'' state. %
However, there are two challenges associated with this approach. %

The first challenge stems from the nature of the symbolic initial
state. %
Since the proof starts from \emph{any} possible state, it also
includes states that are unreachable from reset. %
This can lead to false counterexamples, since the property in question
may hold for the set of reachable states, but fail for certain unreachable states. %
This challenge arises not only in security verification, but in all
formal verification approaches that use a symbolic initial state. %
A standard approach to address this problem is to add \textit{invariants} to the proof: %

\begin{definition}[\textit{Invariant}]
  \label{def:invariant}
  \hfill \\
  For a given HW design, we call a subset of states $B \subseteq S$ an
  \emph{invariant} if for every state $s \in B$ all its
  successor states~$s'$ are also in~$B$. %
  This means that,
  \[ \forall s \in S, \forall i \in I: s \in B \Rightarrow \delta(s, i) \in B \] %
  \qed%
\end{definition}

In other words, an invariant is a set of states that is closed under
reachability. %
To simplify the notation, we implicitly assume that $S_0 \subseteq B$
in the rest of the paper, i.e., we only consider invariants that
include the reset states. %

Even when using a symbolic initial state that ``fast-forwards'' the system to an arbitrary execution state, the data must still propagate from the input through the system
before it can affect a control output. %
Therefore, the second challenge is to handle the structural depth, i.e., the length of the propagation path from $X_D$ to $Y_C$, of the design. %
For large and complex designs, unrolling the circuit is costly and quickly reaches the capacity of formal tools. %
In such cases, it is not possible to make exhaustive statements about
a design's data-obliviousness based on its I/O behavior alone and we
need to look into the internal state of the system. %

For this problem, we propose an \emph{inductive property} for
data-obliviousness that spans a single clock cycle and that also
considers internal state signals. %
Just like for the input/output signals, we partition the set of state variables~$Z$ into
control variables~$Z_C$ and data variables~$Z_D$, where
$Z_C \cup Z_D = Z$ and $Z_C \cap Z_D = \emptyset$. %
Accordingly, $s_C(\tau)$ is the sequence of valuations to~$Z_C$ and
$s_D(\tau)$ is the sequence of valuations to~$Z_D$. %
For the sake of a simplified notation, we also let~$i_C$, $o_C$
and~$s_C$ take an input symbol, output symbol or state of the Mealy
machine and return the valuation of the corresponding subset of
control signals~$X_C$, $Y_C$ and~$Z_C$, respectively. %
We present how to systematically partition~$Z$ into~$Z_C$ and~$Z_D$ later in Sec.~\ref{sec:methodology} such that this process is always conservative in terms of security. %

The data-obliviousness property that we use as an element of our inductive reasoning is shown in
Eq.~\ref{eq:do-1-cycle}. %

\begin{equation}
  \label{eq:do-1-cycle}
  \begin{split}
    \textit{\textbf{DO'}} \coloneq \
    & \exists B \subseteq S, \ \forall s_1, s_2 \in B,\ \forall i_1, i_2 \in I: \\
    & s_C(s_1) = s_C(s_2)  \land i_C(i_1) = i_C(i_2) \\ 
    & \Rightarrow s_C(\delta(s_1, i_1)) = s_C(\delta(s_2, i_2)) \\
    & \land o_C(\lambda(s_1, i_1)) = o_C(\lambda(s_2, i_2))
  \end{split}
\end{equation}
  
This property expresses that if we have two instances of the system for which the control inputs
and states have equal values, then the control state variables will also be equal in the next states of
  the two instances. %
This weakens the initial assumption that a discrepancy between the two
instances can originate only from the data inputs. %
By allowing internal data (state) signals to take arbitrary values, we
implicitly model any propagation of data through the system by the symbolic initial state. %
 
It is important to remember that the invariant~$B$ is a superset of
  the reachable states because we require it to include the initial
  states, $S_0 \subseteq B$. %
In practice, the security property of Eq.~\ref{eq:do-1-cycle} holds
not only in the reachable state set but, often, also in many unreachable states. %
Therefore, finding a suitable invariant for the given property is usually less of a problem than may be generally expected. %
In our proposed methodology (Sec.~\ref{sec:methodology}), we systematically create the necessary invariant in an iterative procedure and prove it on the fly along with the property for data-obliviousness. %

In the same way as for Eq.~\ref{eq:do-1-cycle}, we can derive an
  inductive property for data-obliviousness
  when the set of allowed
  inputs to the design is restricted by a constraint~$\phi$ (Def.~\ref{def:constrained-data-oblivious}). %
This causes a restriction also in the set of reachable states, which
must be expressed by an invariant. %
To this end, we extend Def.~\ref{def:invariant} and denote $B_\phi$ as
a set of states for which
$\forall s \in S, \forall i \in \phi: s \in B_\phi \Rightarrow
\delta(s, i) \in B_\phi$ under a given constraint $\phi$. %
  As an example, assume that the input constraint~$\phi$ excludes branch
  instructions from entering the pipeline of our processor. %
  A corresponding invariant~$B_\phi$ excludes all states that involve
  processing of such branch instructions. %
We elaborate on how to systematically derive such an invariant in Sec.~\ref{sec:methodology}. %

Eq.~\ref{eq:cdo-1-cycle} shows the inductive property for input-constrained data-obliviousness. %
  
\begin{equation}
  \label{eq:cdo-1-cycle}
  \begin{split}
    \textit{\textbf{DO'}}_\phi \coloneq \
    & \exists \phi, \ \exists B_\phi \subseteq S, \ \forall s_1, s_2 \in B_\phi,\ \forall i_1, i_2 \in \phi: \\
    & s_C(s_1) = s_C(s_2) \land i_C(i_1) = i_C(i_2) \\
    & \Rightarrow s_C(\delta(s_1, i_1)) = s_C(\delta(s_2, i_2)) \\
    & \land o_C(\lambda(s_1, i_1)) = o_C(\lambda(s_2, i_2))
  \end{split}
\end{equation}

We now show that, for any \ac{hw} design, our inductive properties
cover their corresponding definitions over infinite traces. %

  \begin{theorem}[(\ref{eq:do-1-cycle}) $\Rightarrow$
    (\ref{eq:data-oblivious}) and (\ref{eq:cdo-1-cycle}) $\Rightarrow$
    (\ref{eq:constrained-data-oblivious})]
    If a \ac{hw} design operates (constrained) data-obliviously
    according to the inductive property
    Eq.~\ref{eq:do-1-cycle}~(Eq.~\ref{eq:cdo-1-cycle}), it also
    operates (constrained) data-obliviously according to the property
    over infinite traces
    Eq.~\ref{eq:data-oblivious}~(Eq.~\ref{eq:constrained-data-oblivious}). %
\end{theorem}

\begin{proof}
  We show that (\ref{eq:cdo-1-cycle}) $\Rightarrow$
  (\ref{eq:constrained-data-oblivious}). This covers (\ref{eq:do-1-cycle}) $\Rightarrow$ (\ref{eq:data-oblivious})
  for the special case that $\phi=I$. %
  \ \\
  We prove this implication by contradiction. %
  Assume a \ac{hw} design fulfills property~(\ref{eq:cdo-1-cycle}) for
  a given $\phi$ but violates property~(\ref{eq:constrained-data-oblivious}), i.e., there exists a
  set of traces $\tau_1, \tau_2 \in T_{\phi}$ with
  $t \in \mathbb{N}_{0}: s(\tau_1[t]) = s(\tau_2[t]) \land
  i_C(\tau_1[t..]) = i_C(\tau_2[t..])$ such that
  $\exists t_k \in \mathbb{N}_{0}: o_C(\tau_1[t_k]) \neq
  o_C(\tau_2[t_k])$ where $t_k > t$. %
  (The case $t_k = t$ can be excluded since
  property~(\ref{eq:cdo-1-cycle}) holds on the design and
  $Z_C \subseteq Z$.) %

$o_C(\tau_1[t_k]) \neq o_C(\tau_2[t_k])$ requires that the antecedent of~(\ref{eq:cdo-1-cycle}) is violated at time point $t_k-1$. %
Since our initial assumption requires $i_C(\tau_1[t_k-1]) = i_C(\tau_2[t_k-1])$, it follows that $s_C(\tau_1[t_k-1]) \neq s_C(\tau_2[t_k-1])$. %
However, this is a contradiction to
\begin{equation*}
\begin{split}
s(\tau_1[t]) = s(\tau_2[t]) \overset{Z_C \subseteq Z}{\Rightarrow} & s_C(\tau_1[t]) = s_C(\tau_2[t]) \\
\overset{(\ref{eq:cdo-1-cycle})}{\Rightarrow} \ \ & s_C(\tau_1[t+1]) = s_C(\tau_2[t+1]) \\
\overset{(\ref{eq:cdo-1-cycle})}{\Rightarrow} ... \overset{(\ref{eq:cdo-1-cycle})}{\Rightarrow} \ \ & s_C(\tau_1[t_k-1]) = s_C(\tau_2[t_k-1])
\end{split}
\end{equation*}%
which means that our initial assumption is wrong, and therefore
(\ref{eq:cdo-1-cycle})~$\Rightarrow$~(\ref{eq:constrained-data-oblivious}). %
 \end{proof}

Proving the relationship "(\ref{eq:cdo-1-cycle}) $\Rightarrow$
(\ref{eq:constrained-data-oblivious})" is crucial to ensure the
validity of our proposed approach. %
If we can verify the inductive property on the design, we
have also verified that the property over infinite traces holds. %
Our approach requires us to find some system invariant for which
the inductive property holds. %
Usually, in practice, such an invariant can include the majority of the
unreachable states. This is key to the feasibility of the proposed methodology, as it makes finding a suitable invariant much more practical. %

The implication in the opposite direction, i.e., "(\ref{eq:constrained-data-oblivious})
$\Rightarrow$ (\ref{eq:cdo-1-cycle})", is not true for the general case. %
The reason for this is that Eq.~\ref{eq:cdo-1-cycle} can
over-approximate the state space, which could lead to false
counterexamples. %
However, this does not affect the validity of the proposed
methodology, as its goal is to prove the inductive property. %
  
\subsection{Special Cases}
\label{subsec:special-cases}

In Def.~\ref{def:data-oblivious}, we model the security-critical
timing behavior by a set of control output signals, e.g., signals
responsible for bus handshaking. %
The concern might arise as to how these properties can be formulated
if the design under verification does not implement any such control
interface at all. %
The question is: %
What is a realistic attacker model in this case? %
If an attacker can observe the data outputs themselves, any
constant-time provisions are futile. %
If we assume a scenario, in which the attacker can spy on
internal signals to determine when an operation has finished, we can
prove data-independence for these signals instead. %
In any case, if there is no handshaking control implemented %
whatsoever, the specification must give detailed information about the timing behavior
of the system. %
Then, timing is an essential part of the circuit's functional
correctness and should be covered by conventional verification
methods. %

\section{Methodology}
\label{sec:methodology}

In this chapter, we present \textit{Unique Program Execution Checking for Data-Independent Timing (UPEC-DIT)} building upon earlier work in~\cite{2022-DeutschmannMueller.etal}. %
UPEC-DIT is a formal methodology to systematically and exhaustively detect data-dependent behavior at the microarchitectural level. %
In particular, we show how UPEC-DIT is used to verify data-obliviousness by proving the properties introduced in the previous chapter (Eq.~\ref{eq:do-1-cycle} and Eq.~\ref{eq:cdo-1-cycle}). %
In the following, for reasons of simplicity, when the distinction between the case of $\phi = I$ and the case of $\phi \subset I$ is irrelevant, we omit the term ''input-constrained'' and simply speak of ''data-obliviousness''. %
 
\subsection{UPEC-DIT Overview}
\label{subsec:upec-dit-overview}

\begin{figure}[ht]
  \centering
  \includegraphics[width=0.85\linewidth]{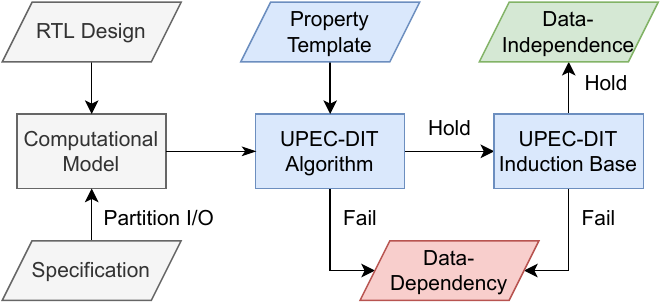}
  \caption{UPEC-DIT Flow}
  \label{fig:upec-dit-flow}
\end{figure}

Fig.~\ref{fig:upec-dit-flow} shows an overview of the different steps
of the methodology. %
We describe each step in more detail in the following subsections. %
We base our work on a methodology called
\textit{\ac{upec}}~\cite{2021-MuellerFadiheh.etal,
  2023-FadihehWezel.etal}. %
\ac{upec} utilizes \ac{ipc}~\cite{2008-NguyenThalmaier.etal,
  2014-UrdahlStoffel.etal} and a 2-safety model to systematically
and exhaustively trace the propagation of confidential information
through the system. %
Originally, \ac{upec} was devised to detect transient execution attacks in processors. %
In that scenario, the secret is stored at a protected location (\emph{data-at-rest}) from which it
must never leak into the architecturally visible
state. %

This paper extends over previous \ac{upec} approaches in that it targets a threat model for \emph{data-in-transit}, i.e., confidential information is processed legally, but must not cause any unwanted side effects. %
This threat model only allows the attacker to observe microarchitectural side effects resulting from the victim's processing of confidential data. %
These side effects take the form of variations in the execution time of the victim's operations. %
UPEC-DIT must therefore tolerate the processing of confidential data within the data path, i.e., it must discriminate between legal and illegal propagation of secrets through the hardware. %
Since these side effects are dictated by the control behavior of the \ac{hw}, the challenge for UPEC-DIT is to detect such secret-dependent control logic alterations. %
  
Our starting point is the \ac{rtl} description of the system. %
Based on the specification, we partition the I/O signals of the design into control and data. %
With this information, we create the computational model and
  initialize the main inductive property for UPEC-DIT which is then
  submitted to the main algorithm that implements the induction step in our global reasoning. %
During the execution of the algorithm, the property is iteratively
refined %
with respect to the internal state signals until it either holds or a counterexample is returned, describing a data-dependent behavior. %
This refinement procedure is conservative in the sense that a wrong decision may lead to a false counterexample, but never to a false security proof. %
If the property holds, the final step is to ensure that our
assumptions for the proof are valid by performing an induction base
proof. %
Once both, the induction step property and the induction base
  property, have been successfully verified we have obtained a formal
  guarantee that the design under verification operates in a
  data-oblivious manner. %
  
The result of our methodology is a set of data-oblivious HW primitives on which SW countermeasures, such as constant-time programming, can rely. %
Given such a certified set of instructions, and provided the \ac{sw} employs such a programming paradigm, data-dependent timing channels are avoided. %
With the same intention, processor vendors like Intel~\cite{web-intel-dit} and ARM~\cite{web-arm-dit} have added information about the data-obliviousness of individual instructions. %
The results of our methodology can be directly matched against such DIT specifications. %
For every instruction of the considered set, our method certifies data-independent timing of its operation and that it leaves no footprints in the design's control state containing information about confidential operands processed. %
Hence, any software composed from the certified set of instructions is data-oblivious. %

\subsection{Computational Model}
\label{subsec:computational-model}

\begin{figure}[ht]
  \centering
  \includegraphics[width=0.95\linewidth]{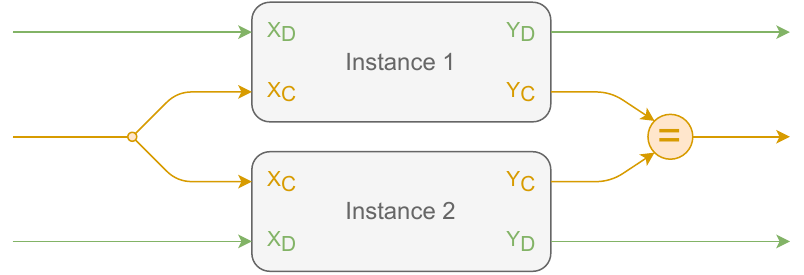}
  \caption{UPEC-DIT Computational Model}
  \label{fig:computational-model}
\end{figure}

Fig.~\ref{fig:computational-model} shows the abstract computational
model for the proposed methodology. %
Like previous \ac{upec} approaches, UPEC-DIT is based on a
2-safety computational model. %
In our model, inputs and outputs of the design are partitioned into
  control and data signals. %
  This is a manual step that, in most cases, is straightforward and
  can be done by consulting the design specification. %
Generally speaking, any confidential information passing through the system must be marked as data. %
After the partitioning, the generation of the 2-safety model can
be fully automated. %
The control inputs~$X_C$ take arbitrary but equal values, whereas
the data inputs~$X_D$ remain unconstrained. %
According to Def.~\ref{def:data-oblivious}, our goal is to prove that
for any sequence of inputs, the control outputs~$Y_C$ never diverge
from their respective counterparts in the other instance. %

\subsection{The UPEC-DIT Property}
\label{subsec:upec-dit-property}

Fig.~\ref{fig:upec-dit-property} shows our \ac{ipc} property template
to formally verify the data-obliviousness of the design. %
 
\begin{figure}[ht]
  \centering
  \begin{minipage}{0.9\linewidth}
    \fontsize{9}{11}\selectfont
    \begin{tabbing}
      XX\=XXXXXX\=XXX\=XXXX\=\kill
      \textbf{UPEC-DIT-Step($Z_C, Y_C, \phi, B_\phi$)}:\\
      \textcolor{blue}{assume:} \\
      \> \textcolor{blue}{at} $t$: \>\>\>\textit{Control\_State\_Equivalence($Z_C$)}; \\
      \> \textcolor{blue}{at} $t$: \>\>\>\textit{Invariants($B_\phi$)}; \\
      \> \textcolor{blue}{during} $t$..$t$+$1$: \>\>\>\textit{Input\_Constraints($\phi$)}; \\
      \textcolor{blue}{prove:} \\
      \> \textcolor{blue}{at} $t$+$1$: \>\>\>\textit{Control\_State\_Equivalence($Z_C$)}; \\
      \> \textcolor{blue}{at} $t$+$1$: \>\>\>\textit{Invariants($B_\phi$)}; \\
      \> \textcolor{blue}{during} $t$..$t$+$1$: \>\>\>\textit{Control\_Output\_Equivalence($Y_C$)}; \\
    \end{tabbing}
  \end{minipage}
  \caption{Interval property template for UPEC-DIT}
  \label{fig:upec-dit-property}
\end{figure} 
It expresses our abstract definitions introduced in
Sec.~\ref{sec:theoretical-foundation} by standard property languages such as \ac{sva}. %
We iteratively refine this property with respect to $Z_C$, $\phi$ and $B_\phi$ during the execution of the
UPEC-DIT algorithm. %
 We now describe the individual components of the property, expressed as \emph{macros} or \emph{functions}, in more detail: %

\begin{itemize}
\item \textit{Control\_State\_Equivalence()} constrains the state
  holding signals related to control ($Z_C$) to be equal in both
  instances of the computational model. %
  At the start of the algorithm, we set $Z_C = Z$, before iteratively
  refining this partitioning. %
  We elaborate on how this is done in
  Sec.~\ref{subsec:upec-dit-algorithm}.%
  
\item \textit{Input\_Constraints()} exclude unwanted behavior to
  achieve input-constrained data-obliviousness (Def.~\ref{def:constrained-data-oblivious}). %
  For systems specifically designed to be data-oblivious, such constraints might not be required. %
  An example for this macro would be ''no new (data-dependent) division issued
  to the processor pipeline''. %
  
\item \textit{Invariants()} are used to constrain the state space to
  exclude unreachable scenarios. %
  These invariants are iteratively deduced during the main algorithm
  (Sec.~\ref{subsec:upec-dit-algorithm}). %
  All invariants used to strengthen our properties are proven inductively \textit{on the fly}, along with data-obliviousness, which is why the macro \textit{Invariants()} is included also in the property commitment. %
   
  \item \textit{Control\_Output\_Equivalence()} is our main proof target and expresses that the outputs marked as control must never diverge. %
\end{itemize}

It is important to note that the control inputs are already
constrained in the computational model itself
(cf.~Fig.~\ref{fig:computational-model}), and, therefore, are not specified in the property. %

\subsection{The UPEC-DIT Algorithm}
\label{subsec:upec-dit-algorithm}

\begin{algorithm}
  \caption{The UPEC-DIT Algorithm}
  \label{alg:upec-dit}
  \begin{algorithmic}[1]
    \Procedure{UPEC-DIT}{}

    \State $Z_C \gets Z$                                                      \label{line:init}
    \State $\phi, B_\phi \gets \emptyset$
    \State \textit{CEX} $\gets $ IPC(UPEC-DIT-Step($Z_C, Y_C, \phi, B_\phi$)) \label{line:check1}
    \While {\textit{CEX} $\neq \{\}$}
    \If {$Y_{C1} \neq Y_{C2} \in$ \textit{CEX}}                               \label{line:output-discrepancy}
    \State \Return \textit{CEX}
    \EndIf
    \For {each $z \in Z_C : z_1 \neq z_2 \in$ \textit{CEX}}                   \label{line:internal-propagation}
    \If {$z$ is $Data$}                                                       \label{line:data-control-decision}
    \State $Z_C = Z_C \backslash \{z\}$                                       \label{line:data-propagation}
    \ElsIf {Invalid \textit{CEX}}
    \State Update($\phi,B_\phi$)                                              \label{line:add-invariant}
    \State $Z_C \gets Z$  
    \State \textit{break}
    \Else
    \State \Return \textit{CEX}                                               \label{line:return-cex}
    \EndIf
    \EndFor
    \State \textit{CEX} $\gets $ IPC(UPEC-DIT-Step($Z_C, Y_C, \phi, B_\phi$)) \label{line:check2}
    \EndWhile
    \State \Return \textit{hold}
    
    \EndProcedure
  \end{algorithmic}
\end{algorithm}

The basic idea of the UPEC-DIT algorithm is to use the counterexamples of the formal tool to iteratively
refine the set of all state signals~$Z$ into control and data subsets~$Z_C$
and~$Z_D$, respectively. %
Leveraging this partitioning of internal signals results in an inductive
proof over only a single clock cycle, which scales even for very large
designs. %

Alg.~\ref{alg:upec-dit} shows the algorithm in pseudocode. %
We begin by initializing the set of control states~$Z_C$ with the set of all state signals~$Z$ in Line~\ref{line:init}. %
In the beginning, the set of input constraints $\phi$ and invariants $B_\phi$ is empty. %
We then call the formal property checker in Line~\ref{line:check1} using the property in Fig.~\ref{fig:upec-dit-property}. %
 In almost all cases, this will produce a counterexample \textit{CEX} which shows a first propagation of data. %

We then continuously investigate the returned counterexamples to
decide on how to proceed: %
If a discrepancy has traveled to a control output
(Line~\ref{line:output-discrepancy}), we detected a data-dependent
timing and return the respective counterexample. %
If a propagation to one or multiple internal signals has occurred
(Line~\ref{line:internal-propagation}), the verification engineer has
to decide if this information flow was valid or not. %
Whenever a propagation to a data signal, e.g., a pipeline buffer, is
detected, we generalize the proof by removing this variable from the set
of control signals (Line~\ref{line:data-propagation}). %
In the case that the propagation touches a signal considered control,
e.g., pipeline \emph{stall} or \emph{valid} signals, the algorithm concludes and
returns the counterexample (Line~\ref{line:return-cex}). %

In some cases, the counterexamples may show behavior that is either
unreachable or invalid in the given application
  scenario. %
An example for an invalid behavior could be the exclusion of certain
instruction types for which data-obliviousness is not required, e.g.,
branches. %
To handle these cases, invariants or input constraints must be added to restrict the set of
  considered states (Line~\ref{line:add-invariant}). %
Afterwards, we also reset the set of control signals~$Z_C$. %
This is required for correctness, as an added input constraint
might also make some previously considered propagation impossible. %
In our experiments, however, resetting~$Z_C$ did not occur often and
  thus did not cause significant overhead. 
After every discrepancy in the counterexample has been investigated,
we rerun the proof with the new assumptions
(Line~\ref{line:check2}). %
 This is continued until the proof holds and no new counterexample can be found by the
formal property checker. %
In this case, the algorithm terminates. %
  We continue by verifying the
induction base, i.e., we check, as described in the following subsection, that our initial assumptions hold after
a system reset. %

We emphasize that the proposed property covers any interdependence between elements of a \emph{sequence} of HW operations. %
If such a sequence exists, the first operation produces a differing starting state for the later operation. %
This manifests itself as a discrepancy in either the data or control state space: %
A difference in the control state, however, violates the commitment of our property and leads to a counterexample. %
A discrepancy of data is already taken into account for any subsequent operation, since we assume that all data signals can take arbitrary values. %
As a result, UPEC-DIT never misses a timing channel that requires a specific combination of operations. %

In addition, is important to note that the proposed methodology is
  conservative in the sense that if a control-related signal from~$Z$
  is mistakenly declared as data
  (Line~\ref{line:data-control-decision}), it will not result in a
  false security proof. %
  In this case, the algorithm will continue until the propagation
  eventually reaches a control output~$Y_C$ and returns the
  corresponding counterexample. %
  The main purpose of the individual partitioning of the state set
  signals is to detect and stop such propagation as early as
  possible. %
 
\subsection{Induction Base}
\label{subsec:induction-base}

By successfully proving the inductive property (Fig.~\ref{fig:upec-dit-property}), we show that our 2-safety model never diverges in its control behavior during operation. %
For completing our proof of data-obliviousness, we also have to verify that the assumptions and invariants made in this inductive proof are correct. %
Therefore, the last step of the methodology is to prove the induction base, i.e., that the system starts in a data-oblivious state. %

\begin{figure}[ht]
  \centering
  \begin{minipage}{0.9\linewidth}
    \fontsize{9}{11}\selectfont
    \begin{tabbing}
      XX\=XXXXXX\=XXX\=XXXX\=\kill
      \textbf{UPEC-DIT-Base($Z_C, Y_C, \phi, B_\phi$)}:\\
      \textcolor{blue}{assume:} \\
      \> \textcolor{blue}{during} $t$-$k$..$t$-$1$: \>\>\>\textit{Reset\_Sequence()}; \\
      \> \textcolor{blue}{during} $t$-$k$..$t$: \>\>\>\textit{Input\_Constraints($\phi$)}; \\
      \textcolor{blue}{prove:} \\
      \> \textcolor{blue}{at} $t$: \>\>\>\textit{Control\_State\_Equivalence($Z_C$)}; \\
      \> \textcolor{blue}{at} $t$: \>\>\>\textit{Invariants($B_\phi$)}; \\
      \> \textcolor{blue}{at} $t$: \>\>\>\textit{Control\_Output\_Equivalence($Y_C$)}; \\
    \end{tabbing}
  \end{minipage}
  \caption{Interval property template for UPEC-DIT Induction Base}
  \label{fig:upec-dit-base}
\end{figure} 

Fig.~\ref{fig:upec-dit-base} shows our IPC property template for the
induction base. %
We want to show that the reachability assertions~$B_\phi$ we introduced during the iterative algorithm include the reset state of the system, which means that they are indeed invariants according to Def.~\ref{def:invariant}. %
Furthermore, we want to prove that the system is properly initialized
regarding its control behavior. %
In essence, \textit{Control\_State\_Equivalence()} ensures that all
control state signals~$Z_C$ are initialized. %
If this commitment fails, then the system could behave differently
after a reset, which could hint to a functional bug. %
Lastly, by assuming \textit{Control\_Output\_Equivalence()}, we verify
that there is no combinatorial path from a data input~$X_D$ to a
control output~$Y_C$. %

When both the base and step property hold, we have exhaustively
verified that our design operates data-obliviously, either in the
strong sense of Def.~\ref{def:data-oblivious}, or for a subset of its
total behavior in the weakened sense of
Def.~\ref{def:constrained-data-oblivious}. %
This represents an \emph{unbounded} formal proof that, due to its inductive
nature, can scale to very large designs. %

\section{Optimizations}
\label{sec:optimizations}

The methodology described Sec.~\ref{sec:methodology} verifies data
  oblivious-behavior at the microarchitectural level exhaustively. %
In practice, it may not always be necessary to be exhaustive and an efficient
bug-hunting may suffice for the intended application of some
designs. %
In addition, the low complexity of certain designs allows for
exhaustive verification of data-obliviousness at the I/O interface
level, without the need for invariants and constraints on the internal
behavior of the design. %
To this end, we now discuss some
enhancements and trade-off techniques that may be useful for applying
UPEC-DIT in practice. %

\subsection{Unrolled Proofs}
\label{subsec:unrolled-proofs}

The methodology presented in Sec.~\ref{sec:methodology} uses
an inductive proof with a single-cycle property to avoid complexity
issues. %
In this approach, the set of all possible data propagation paths is
over-approximated in the symbolic initial state by leaving
the values of internal data signals unconstrained. %
While this over-approximation leads to a very low proof complexity, it
implies the need to deal with the possibility of spurious counterexamples. %
Writing invariants can overcome this problem, but in some cases it is
affordable to simply increase the computational effort to avoid these problems. %
If the complexity of the system allows for a sufficient number of unrollings in our computational model, %
considering the full propagation path starting from
any data input to any control output can significantly reduce the number of false counterexamples
and thus the effort of writing invariants. %
 This unrolled approach represents the original UPEC-DIT methodology, as
described in~\cite{2022-DeutschmannMueller.etal}. %

\begin{figure}[ht]
  \centering
  \begin{minipage}{0.9\linewidth}
    \fontsize{9}{11}\selectfont
    \begin{tabbing}
      XX\=XXXXXX\=XXX\=XXXX\=\kill
      \textbf{UPEC-DIT-Unrolled-IO($Z, Y_C, \phi$)}:\\
      \textcolor{blue}{assume:} \\
      \> \textcolor{blue}{at} $t$: \>\>\>\textit{State\_Equivalence($Z$)}; \\
      \> \textcolor{blue}{at} $t$: \>\>\>\textit{Input\_Constraints($\phi$)}; \\
      \textcolor{blue}{prove:} \\
      \> \textcolor{blue}{during} $t$..$t$+$k$: \>\>\>\textit{Control\_Output\_Equivalence($Y_C$)}; \\
    \end{tabbing}
  \end{minipage}
  \caption{Unrolled UPEC-DIT property only considering I/O-behavior}
  \label{fig:upec-dit-unrolled-property-io}
\end{figure}

The idea of unrolled proofs is shown in
Fig.~\ref{fig:upec-dit-unrolled-property-io} and is a straightforward
implementation of Def.~\ref{def:data-oblivious}. %
In this property, \textit{all} state signals~$Z$ are initialized to
equal but arbitrary values between the two instances. %
This is denoted by the \textit{State\_Equivalence()} macro. %
We then prove that for a maximum latency~$k$, the two instances
maintain equal control outputs~$Y_C$. %
We choose~$k$ to be greater or equal to the length of the longest
\ac{hw} operation in the design. %
If this property fails, it means that the difference in~$Y_C$ must
originate from the data inputs~$X_D$, since this is the only source of
discrepancy between the two instances. %
In this case, the property checker returns a counterexample which
guides the verification to the root cause by highlighting the
deviating values. %

The great advantage of this variant of UPEC-DIT is that it does not
require an iterative partitioning of internal state signals~$Z$. %
Therefore, the only manual steps are partitioning the system interface
and choosing a maximum latency~$k$. %
Everything else can be automated. %
  For many low-complexity designs, such as functional units or
accelerators, this approach can provide exhaustive proofs. %
It can also serve as a quick initial test for larger systems, as most
timing channels become visible after only a few cycles. %
Unfortunately, this approach can run into scalability problems
because the full propagation paths from input to output can be too long
in more complex systems such as processor cores. %

\begin{figure}[ht]
  \centering
  \begin{minipage}{0.9\linewidth}
    \fontsize{9}{11}\selectfont
    \begin{tabbing}
      XX\=XXXXXX\=XXX\=XXXX\=\kill
      \textbf{UPEC-DIT-Unrolled($Z, Z_C, Y_C, \phi$)}:\\
      \textcolor{blue}{assume:} \\
      \> \textcolor{blue}{at} $t$: \>\>\>\textit{State\_Equivalence($Z$)}; \\
      \> \textcolor{blue}{at} $t$: \>\>\>\textit{Input\_Constraints($\phi$)}; \\
      \textcolor{blue}{prove:} \\
      \> \textcolor{blue}{during} $t$..$t$+$k$: \>\>\>\textit{Control\_State\_Equivalence($Z_C$)}; \\
      \> \textcolor{blue}{during} $t$..$t$+$k$: \>\>\>\textit{Control\_Output\_Equivalence($Y_C$)}; \\
    \end{tabbing}
  \end{minipage}
  \caption{Unrolled UPEC-DIT property with internal state signals}
  \label{fig:unrolled-upec-dit-property-states}
\end{figure} 

A trade-off between computational complexity and a decreased number of false counterexamples is presented in Fig.~\ref{fig:unrolled-upec-dit-property-states}. %
This variant of UPEC-DIT also starts by initializing \textit{all} state signals~$Z$ to equal but arbitrary values, and thus considers propagation paths starting from the data inputs. %
In this case however, we perform a partitioning of~$Z$ into~$Z_C$ and~$Z_D$ based on Alg.~\ref{alg:upec-dit}. %
Having an internal representation of the control-flow allows for a much earlier detection of data-dependent side-effects. %
Furthermore, isolating the source of the discrepancy to the input makes the returned counterexamples more intuitive and less likely to be spurious. %
 
Unfortunately, this approach does not scale well for complex systems beyond a few clock cycles. %
Nonetheless, this unrolled method can serve as a basis for the inductive proof, since the deduced partitioning of~$Z$ is the same for both variants. %
Therefore, it often makes sense to start with the unrolled approach and set~$k=1$. %
The verification engineer can then iteratively increase~$k$ until no new counterexamples appear or the computational cost becomes prohibitive. %
Then the transition to the inductive method is made by initializing~$Z_C$ in Line~\ref{line:init} of Alg.~\ref{alg:upec-dit} with the remaining control state signals. %
In our experience, starting with an unrolled model makes the counterexamples more intuitive because it shows longer propagation paths starting from the inputs. %
 This can be especially helpful early in the methodology, when the verification engineer has little to no knowledge or intuition about the design under verification. %
We have omitted the invariants in Fig.~\ref{fig:upec-dit-unrolled-property-io} and Fig.~\ref{fig:unrolled-upec-dit-property-states} because they are usually very simple when all state signals $Z$ are initialized to equal values. \vspace{1ex}%
  \par Finally, we would like to point out that the unrolled property
  (Fig. 6) can also be used in an effective \emph{bug hunting} approach that trades formal/exhaustive coverage for efficiency. %
  Instead of executing the full iterative algorithm that refines the
  set of state signals systematically (required for formal
  exhaustiveness), we can specify
  \textit{Control\_State\_Equivalence()} using a set of state signals
  identified as control manually, and run the proof. %
  The verification engineer determines the control variables based on
  knowledge and experience. %
  Obvious examples of control signals are stall variables in a
  processor pipeline. %
  Empirical evidence from our experiments shows that almost all timing
  vulnerabilities manifest themselves in only a handful of control
  signals. %
  While this short-cut over the formal algorithm bears a certain risk
  of missing corner-case behavior, it avoids the potentially laborious
  iterative procedure and produces high-quality results fast. %
  It may
  be a viable option in certain practical scenarios. %

\subsection{Black-Boxing}
\label{subsec:black-boxing}

Black-boxing can significantly increase the scalability of the formal
proof. %
Essentially, black-boxing excludes the functionality of certain
components from the system, reducing the complexity of the
computational model. %
When a module is black-boxed, its inputs become outputs of the system
under verification. %
Likewise, any output of the module connected to the rest of the system
is now considered as an open input, since the functionality of that
component is no longer considered. %
In particular, for state-heavy submodules, such as caches, their
black-boxing can greatly simplify the overall state space for the
formal proof. %

\begin{figure}[ht]
  \centering
  \includegraphics[width=0.95\linewidth]{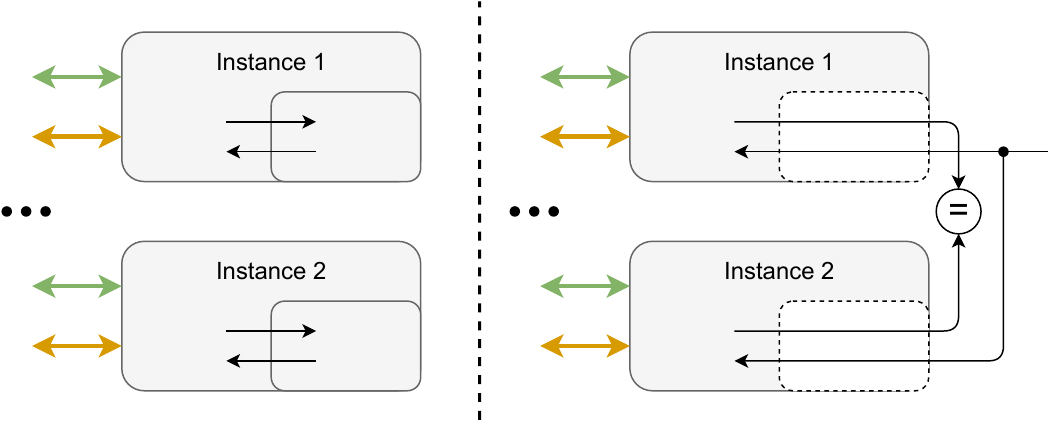}
  \caption{Computational Model before and after sound black-boxing}
  \label{fig:black-boxing}
\end{figure}

Abstraction using black-boxing is widely supported by commercial and academic formal tools. %
However, conventional formal verification often requires an abstract representation of the black-boxed component to model the behavior partially at the interface between the component and the rest of the system. %
Over- or under-approximation of such an abstract representation can lead to false counterexamples or verification gaps, respectively. %
Fortunately, the 2-safety computational model of UPEC-DIT allows for \textit{sound} black-boxing, ensuring that no security violations are missed. %
We only need to monitor the interface of the black-boxed component, as shown in Fig.~\ref{fig:black-boxing}. %
Moreover, false counterexamples can be avoided by assuming that the outputs of the component produce equal values in the two instances, as long as no discrepancy propagates to the inputs of the component. %

If a counterexample is produced that shows a difference at the
black-box inputs, the verification engineer can decide to either undo
the black-boxing or examine the module individually. %
The first option requires less effort but results in higher computational complexity, while the second option requires more manual
effort but can lead to a better understanding of the system and
simpler counterexamples. %
We will explore the second option further in the following
subsection. %

\subsection{Modularization}

Setting up formal proofs from scratch for a very large system can be a
difficult task. %
Therefore, it is often advisable to decompose the problem and to first
look at individual components that are "suspicious" or critical to
security. %
Examples would be a cryptographic accelerator or, in the case of a
processor, the various functional units of the pipeline. %

Investigating an individual component results in a less complex computational model, simpler counterexamples, and helps to establish a better understanding of the system. %
We can utilize the same inductive approach as elaborated in Sec.~\ref{sec:methodology} or, if computational complexity permits, the unrolled approach described in Sec.~\ref{subsec:unrolled-proofs}. %
If a counterexample is found for a single module, it is very likely that it also indicates a security threat to the entire system. %
If a component turns out to be data-oblivious, we can use this information to simplify our computational model of the system using black-boxing~(Sec.~\ref{subsec:black-boxing}). %
For this purpose, we consider the control (data) inputs of the black-box as control (data) outputs of the system and the control (data) outputs of the black-box as control (data) inputs of the system. %
Therefore, we can skip the data propagation through the black box and use its data output as a new source of discrepancy. %
This allows us to systematically partition the formal proof in a divide-and-conquer fashion, making it scalable even for very large systems. %

\subsection{Further Optimizations}
	
Two other well-known optimization
concepts can be applied to UPEC-DIT. %
First, each property can be divided into several independent, parallelizable properties that
check whether propagation into a particular signal is possible. %
This means that, instead of verifying all possible
propagation paths at once, we generate and
verify an individual property for each~$z \in Z_C$. %
Second, a \emph{cone-of-influence} reduction can be used to reduce
the number of candidates in the property commitment. %
Instead of checking for
\textit{Control\_State\_Equivalence()}, we can check only for the
equivalence of all state signals that are in the fan-out of
$I \cup Z \backslash Z_C$. %

\section{Experiments}
\label{sec:experiments}

All experiments are available in our GitHub repository~\cite{web-upec-dit}. %

\subsection{Example: SHA512 Core}

  We begin our practical evaluation with a practical example
    demonstrating the proposed methodology in detail, namely an open-source implementation of a
  cryptographic accelerator implementing the SHA512
  algorithm~\cite{web-opencores-sha}. %
  A timing channel in such a core can create severe security flaws, as
  it can significantly reduce the strength of the underlying
  encryption. %
  Therefore, we want to exhaustively verify that no data-dependent
  timing behavior exists in this accelerator. %

We begin by looking at the interface of the module along with its specification. %
The core implements 5 inputs (\textit{clk\_i}, \textit{rst\_i}, \textit{text\_i}, \textit{cmd\_i} and \textit{cmd\_w\_i}) and~2 outputs (\textit{text\_o} and \textit{cmd\_o}). %
After referring to the documentation, we mark the clock, reset and handshaking signals (\textit{cmd}) as \textit{control}, while the plain and cipher text ports are marked as \textit{data}. %
Our goal is to prove that the accelerator is data-oblivious according to Def.~\ref{def:data-oblivious}, i.e., the data input \textit{text\_i} has no influence on the control output \textit{cmd\_o}. %

The next step is to generate our 2-safety computational model (Fig.~\ref{fig:computational-model}), in which the control inputs are constrained to be equal, while the data input remains unconstrained. %
We also generate a macro for \textit{State\_Equivalence()} that constrains each of the 37 state-holding signals (2162 flip-flops) to equal values between the two instances at the start of our proof. %
Given the I/O partitioning, both steps can be fully automated with dedicated tool support. %

We can now choose to formulate either a single-cycle or an unrolled proof. %
For an exhaustive proof over multiple clock cycles, we have to unroll the model to its sequential depth, as elaborated in Sec.~\ref{subsec:unrolled-proofs}. %
According to the specification, one encryption operation has a latency of 97 clock cycles. %
Since the complexity of this module is rather low, it is still possible to unroll the formal proof for several hundred cycles. %
For more complex accelerators, this is usually not feasible. %

Therefore, to show the effectiveness of the UPEC-DIT methodology, we
decide to iteratively create an inductive single-cycle proof using
Alg.~\ref{alg:upec-dit}. %
We begin by setting up %
the UPEC-DIT-Step property (Fig.~\ref{fig:upec-dit-property} in Sec.~\ref{subsec:upec-dit-property}) %
in such a way that \textit{Control\_State\_Equality()} considers all
37 state holding signals~$Z$ %
of the design and \textit{Control\_Output\_Equality()} ensures the
equality of \textit{cmd\_o}. %
After running the verification tool, we receive a first counterexample
showing a discrepancy propagation to %
two internal pipeline buffers (\textit{W3} and \textit{Wt}). %
Since these are used to store intermediate results of the encryption,
we classify them as \textit{data}, exclude them from the macro and rerun the
proof. %
 
After several iterations of checking the UPEC-DIT-Step property%
, a fixed point is reached with only 5 out of~37 signals left in~$Z_C$: %
\textit{busy}, \textit{cmd}, \textit{read\_counter}, \textit{round}
and \textit{Kt}. %
All other signals only store intermediate results and were therefore
marked as data and removed from the proof. %
This means that even though a vast majority of the design can be in an
arbitrary state, the timing behavior only depends on certain control
registers inside the design. %

The last step is to verify the induction base with %
the UPEC-DIT-Base property (Fig.~\ref{fig:upec-dit-base} in Sec.~\ref{subsec:induction-base}), which also holds. %
Hence, we successfully showed that the core operates
timing-independent w.r.t.\ its input data. %
With some experience, the entire proof procedure can be completed in
less than an hour. %

\subsection{Functional Units and Accelerators}
\label{subsec:fu-and-accelerators}

\begin{table}
  \small
  \begin{center}
    \caption{Results for functional units}
    \label{tab:functional-unit-results}
    \vspace{-1ex}
    \begin{tabular}{lcrrr}
      \hline &&&& \\ [-2ex]
      \hline &&&& \\ [-2ex]
      \textbf{Design} & \textbf{DIT} & \textbf{State bits}  & \textbf{Time} & \textbf{Mem} \\ 
      \hline &&&& \\ [-2ex]
      BasicRSA         & \textcolor{red}{\xmark}          & 532  & \textless 1s &  589 MB \\
      SHA1             & \textcolor{ColorCmark}{\cmark}   & 911  & \textless 1s &  306 MB \\
      SHA256           & \textcolor{ColorCmark}{\cmark}   & 1103 & \textless 1s &  296 MB \\
      SHA512           & \textcolor{ColorCmark}{\cmark}   & 2162 & \textless 1s &  329 MB \\
      AES (secworks)   & \textcolor{ColorCmark}{\cmark}   & 2472 & 4s           &  994 MB \\
      AES (OpenCores)  & \textcolor{ColorCmark}{\cmark}   & 554  & \textless 1s &  819 MB \\
      FWRISCV MDS-Unit & \textcolor{orange}{\textbf{(!)}} & 331  & \textless 1s &  596 MB \\
      ZipCPU Div-Unit  & \textcolor{red}{\xmark}          & 142  & 11s          & 1347 MB \\
      CVA6 Div-Unit    & \textcolor{red}{\xmark}          & 209  & \textless 1s &  580 MB \\
      \hline &&&& \\ [-2ex]
      \hline &&&& \\ [-2ex]
    \end{tabular}
  \end{center}
  \begin{minipage}{1.0\linewidth}
    For each entry, we list whether or not the design has data-independent timing
    (DIT), number of state bits, total proof time and peak
    memory usage. %
  \end{minipage}
  \vspace{-3ex}
\end{table}

Tab.~\ref{tab:functional-unit-results} shows the results for several \acp{fu} and accelerators. %
All of these experiments were conducted using the unrolled approach as elaborated in Sec.~\ref{subsec:unrolled-proofs}. %

The first design, \textit{BasicRSA}, was taken from \textit{OpenCores}~\cite{web-opencores-basicrsa}. %
It implements an accelerator for the RSA cryptosystem. %
This module computes the modular exponentiation needed for the encryption algorithm in a \textit{square-and-multiply} fashion. %
While this approach is very efficient, it makes the latency directly dependent on the number of 1s in the exponent, i.e., the value of the key. %
In addition, the timing of each individual multiplication depends on the intermediate results, since the submodule terminates earlier for smaller multipliers. %
Therefore, the overall timing depends not only on the exponent, but also on the base and modulus operands. %
In our experiment, we split the proof into three properties to show the dependence of each data input separately. %

\textit{SHA1}, \textit{SHA256} and \textit{SHA512}~\cite{web-opencores-sha} implement accelerators for standardized variants of the \textit{Secure Hash Algorithm}, which differ in the number of rounds and in the length of their message digests. %
We also looked at two different accelerator implementations~\cite{web-secworks-aes, web-opencores-aes} of the \textit{Advanced Encryption Standard (AES)}. %
These symmetric cryptosystems are naturally quite resilient to timing side channels due to their round-based algorithms. %
However, although unlikely in practice, the microarchitecture of these accelerators could still implement data-dependent side-effects. %
All of these accelerators were proven by UPEC-DIT to have data-independent timing, which met our initial expectations. %

The \textit{Multiplication-Division-Shifting-Unit} is part of the \textit{Featherweight \RISCV{}}~\cite{web-fwriscv} project. %
Its goal is to build a resource-efficient implementation for FPGAs. %
All of its operations take multiple clock cycles, shifting only one bit in each cycle. %
In this \ac{fu}, multiplication and division were proven to execute data-independently, requiring 33 cycles to complete. %
However, UPEC-DIT produced a counterexample for shift operations, as the timing is directly dependent on the shift amount. %
This violates the common assumption of constant-time programming that shift operations are data-oblivious. %
The reason for this is that most modern processors employ \textit{barrel shifters}, i.e., they are able to perform an $N$-bit shift in a single clock cycle. %

The last two designs are \textit{serial division units} taken from the \textit{ZipCPU}~\cite{web-zipcpu} and \textit{CVA6}~\cite{2019-ZarubaBenini} open-source projects. %
Both \ac{fu}s showed strong dependencies of their timing w.r.t.\ their operands. %
\textit{ZipCPU} implements an early termination when dividing by zero. %
In \textit{CVA6}, the \ac{fu} is optimized to perform as few subtractions as possible. %
At the start of the division, the number of subtractions required is computed based on the operand values. %
Essentially, the greater the difference in size between the two operands, the more subtractions must be performed. %
As a result, the latency of the \ac{fu} is strongly data-dependent. %

In all case studies, UPEC-DIT was applied without \textit{a~priori} knowledge of the designs. %
It systematically guided the user to the points of interest. %
Even though the designs have operations taking up to~192~cycles, proof time and memory requirements remained insignificant. %
Furthermore, operations like multiplication, which are usually a bottleneck for formal tools, did not cause any complexity issues. %
This is by merit of the proposed 2-instance computational model which abstracts from functional signal valuations and only considers the difference between the two instances. %

\subsection{In-Order Processor Cores}

\begin{table}
  \small
  \begin{center}
    \caption{Results for in-order processors}
    \label{tab:in-order-processor-results}
    \vspace{-1ex}
    \begin{tabular}{lrrrrr}
      \hline &&&&& \\ [-2ex]
      \hline &&&&& \\ [-2ex]
      & \textbf{FW} & \textbf{IBEX} & (+DIT) & \textbf{SCARV} & \textbf{CVA6} \\ 
      \hline &&&&& \\ [-2ex]
      I-Type  & \textcolor{ColorCmark}{\cmark} %fwriscv                      
      & \textcolor{ColorCmark}{\cmark} %ibex (no dit)  
      & \textcolor{ColorCmark}{\cmark} %ibex (dit)     
      & \textcolor{ColorCmark}{\cmark} %scarv   
      & \textcolor{ColorCmark}{\cmark} %cva6
      \\
      R-Type  & \textcolor{ColorCmark}{\cmark}/\textcolor{red}{\xmark} %fwriscv     
      & \textcolor{ColorCmark}{\cmark}/\textcolor{ColorCmark}{\cmark} %ibex (no dit) 
      & \textcolor{ColorCmark}{\cmark}/\textcolor{ColorCmark}{\cmark} %ibex (dit)  
      & \textcolor{ColorCmark}{\cmark}/\textcolor{ColorCmark}{\cmark} %scarv 
      & \textcolor{ColorCmark}{\cmark}/\textcolor{ColorCmark}{\cmark} %cva6 
      \\
      Mult    & \textcolor{ColorCmark}{\cmark}/\textcolor{ColorCmark}{\cmark} %fwriscv   
      & \textcolor{ColorCmark}{\cmark}/\textcolor{red}{\xmark} %ibex (no dit) 
      & \textcolor{ColorCmark}{\cmark}/\textcolor{ColorCmark}{\cmark} %ibex (dit)  
      & \textcolor{ColorCmark}{\cmark}/\textcolor{ColorCmark}{\cmark} %scarv 
      & \textcolor{ColorCmark}{\cmark}/\textcolor{ColorCmark}{\cmark} %cva6 
      \\
      Div     & \textcolor{ColorCmark}{\cmark}/\textcolor{ColorCmark}{\cmark} %fwriscv   
      & \textcolor{ColorCmark}{\cmark}/\textcolor{red}{\xmark} %ibex (no dit) 
      & \textcolor{ColorCmark}{\cmark}/\textcolor{ColorCmark}{\cmark} %ibex (dit)  
      & \textcolor{ColorCmark}{\cmark}/\textcolor{ColorCmark}{\cmark} %scarv 
      & \textcolor{red}{\xmark}/\textcolor{red}{\xmark} %cva6 
      \\
      Load    & \textcolor{orange}{\textbf{(!)}} %fwriscv                           
      & \textcolor{orange}{\textbf{(!)}} %ibex (no dit)   
      & \textcolor{orange}{\textbf{(!)}} %ibex (dit)    
      & \textcolor{red}{\xmark} %scarv   
      & \textcolor{red}{\xmark} %cva6 
      \\
      Store   & \textcolor{orange}{\textbf{(!)}}/\textcolor{ColorCmark}{\cmark} %fwriscv 
      & \textcolor{orange}{\textbf{(!)}}/\textcolor{ColorCmark}{\cmark} %ibex (no dit) 
      & \textcolor{orange}{\textbf{(!)}}/\textcolor{ColorCmark}{\cmark} %ibex (dit)  
      & \textcolor{red}{\xmark}/\textcolor{red}{\xmark} %scarv 
      & \textcolor{red}{\xmark}/\textcolor{ColorCmark}{\cmark} %cva6 
      \\
      Jump    & \textcolor{orange}{\textbf{(!)}} %fwriscv                                                     
      & \textcolor{ColorCmark}{\cmark} %ibex (no dit)   
      & \textcolor{ColorCmark}{\cmark} %ibex (dit)    
      & \textcolor{orange}{\textbf{(!)}} %scarv   
      & \textcolor{red}{\xmark} %cva6 
      \\
      Branch  & \textcolor{red}{\xmark}/\textcolor{red}{\xmark} %fwriscv       
      & \textcolor{red}{\xmark}/\textcolor{red}{\xmark} %ibex (no dit) 
      & \textcolor{ColorCmark}{\cmark}/\textcolor{ColorCmark}{\cmark} %ibex (dit)  
      & \textcolor{red}{\xmark}/\textcolor{red}{\xmark} %scarv 
      & \textcolor{red}{\xmark}/\textcolor{red}{\xmark} %cva6 
      \\
      \hline &&&&& \\ [-2ex]
      State bits & 3086 & 1019 & 1021 & 2334 &  682849 \\
      AT         & 0:03 & 2:06 & 4:24 & 3:27 & 1:35:41 \\
      WCT        & 0:04 & 5:11 & 6:40 & 8:02 & 3:06:43 \\
      Mem        &  1.7 &  4.5 &  4.3 &  2.1 &    11.9 \\
      \hline &&&&& \\ [-2ex]
      \hline &&&&& \\ [-2ex]
    \end{tabular}
  \end{center}
  \begin{minipage}{1.0\linewidth}
    In each experiment (separate proof), UPEC-DIT proved that the instruction class
    executes data-independently either
    always~(\textcolor{ColorCmark}{\cmark}), only under certain SW
    restrictions~(\textcolor{orange}{\textbf{(!)}}), or not at all, i.e, it
    depends on its operands (\textcolor{red}{\xmark}). %
    Multi-operand instructions denote rs1/rs2 separately, as they can
    have different impact on timing. %
    We also report the number of state bits in the original design,
    the \textit{average time (AT)} and \textit{worst-case time (WCT)}
    of a single proof (HH:MM:SS) as well as the peak memory
    requirements (GB). %
  \end{minipage}
  \vspace{-2ex}
\end{table}

We investigated four different open-source in-order \RISCV{} processors from
low to medium complexity, as shown in Tab.~\ref{tab:in-order-processor-results}. %
The \textit{Ibex} processor is listed twice, as it comes with a
\textit{data-independent timing (DIT)} security feature, which we
examined separately. %
All of these experiments were conducted using the unrolled approach as elaborated in Sec.~\ref{subsec:unrolled-proofs}. %
The results show that time and memory requirements are still moderate, even in the case of a medium-sized processor. %

The first design we investigated is the sequential
\textit{Featherweight \RISCV{}}~\cite{web-fwriscv} processor which
aims at balancing performance with FPGA resource utilization. %
As our results show, most instructions execute independently of their
input data. %
However, there was one big exception, namely, R-Type shift
instructions. %
For area efficiency, the implementation shares a single shifting unit
for multiplication, division and shifting. %
The shifting unit can only shift one bit in each cycle (cf.~\ref{subsec:fu-and-accelerators}), which results in
data-dependent timing depending on the shift amount (rs2). %
We singled out the shift instructions in a separate proof and showed
that other R-Type instructions like addition do, in fact, preserve
data-independent timing. %
Note that I-Type shift instructions also execute with dependence on
the shift amount. %
The shift amount, however, is specified in the (public) immediate
field of the instruction. %
Consequently, since the program itself is viewed as \textit{public},
I-Type shifting executes data-obliviously. %
Load, Store and Jump (JALR) can cause an exception in case of a
misaligned address, while Branches incur a penalty if a branch is not taken. %

The \textit{Ibex \RISCV{} Core}~\cite{web-ibex} is an extensively
verified, production-quality open-source 32-bit \RISCV{} CPU. %
It is maintained by the \textit{lowRISC} not-for-profit company and
deployed in the OpenTitan platform. %
 It is highly configurable and comes with a variety of security
features, including a %
\textit{data-independent timing (DIT)} mode. %
When activating this mode during runtime, execution times of all
instructions are supposed to be independent of input data. %
In our experiments, we apply
UPEC-DIT for both inactive and active DIT mode and use the default \textit{"small"}
configuration, with the \textit{"slow"} option for multiplication. %

When the DIT mode is turned off, we found three cases of
data-dependent execution time: %
\begin{itemize}
  \item Division and (slow) multiplication implement fast paths for
    certain scenarios. %
  \item Taken branches cause a timing penalty, as the prefetch buffer
    has to be flushed. %
  \item Misaligned loads and stores are split into two aligned memory
    accesses. % 
\end{itemize}

The first two issues are solved when DIT mode is active, as seen in
Tab.~\ref{tab:in-order-processor-results}. %
All fast paths are deactivated and non-taken branches now introduce a
delay to equal the timing of taken branches. %
However, the timing violation for misaligned memory accesses is not
addressed. %

When running Ibex in DIT mode, data-oblivious memory accesses require
special measures, such as the integration of the core with a
data-oblivious memory sub-system. %
For example, an oblivious RAM controller~\cite{2015-FletcherRen.etal}
makes any memory access pattern computationally indistinguishable from
any other access pattern of the same length. %
However, our experiments with UPEC-DIT reveal that even with such
strong countermeasures in place, Ibex still suffers from a side
channel in the case of memory accesses that are misaligned. %
This is because the core creates a different number of memory requests
for aligned and misaligned accesses. %
We reported this issue to the lowRISC team and suggested to disable
the misaligned access feature for DIT mode. %
With this fix, the HW would remain secure even in case that a
faulty/malicious SW introduces a misaligned access. %
The lowRISC team refined the documentation and will consider the
proposed fix for future updates of the core. %

The \textit{SCARV}~\cite{web-scarv} is a 5-stage single-issue in-order
CPU core, implementing the \RISCV{} RV32IMC instruction sets. %
We prove that most instructions do not leak information through timing
by running UPEC-DIT on the core. %
Taken branches, however, use additional cycles due to a pipeline flush. %
Memory access instructions cause exceptions if their addresses are misaligned. %
To our surprise, however, not only the address but also the \textit{value} of a store instruction can cause data-dependent behavior. %
The reason for this is that loads and stores to a particular address region are interpreted as memory-mapped I/O accesses and do not issue a memory request. %
This can be used to set a \ac{sw} interrupt timer depending on the store value. %
As a result, the timing of such an interrupt is directly correlated to the data of the store instruction. %

\textit{CVA6}~\cite{2019-ZarubaBenini}, also known as \textit{Ariane},
Having almost 700\,k state bits, the design itself already pushes the limits of formal property checkers. %
To make things worse, a straightforward 2-safety circuit model would
have twice as many state bits. %
Fortunately, UPEC-DIT allows for \textit{sound} black-boxing to cope
with complexity issues~(cf.~Sec.~\ref{subsec:black-boxing}). %
Black-boxing reduces the computational model to~24\,k state bits. %
For load and store instructions, a couple of exception scenarios
showed up: %
Access faults (PMP), misaligned addresses or page faults. %
Besides a data-dependent division and timing variations by
mispredicted branches UPEC-DIT also found an interesting case in which
a load can be delayed. %
In order to prevent RAW hazards, whenever a load is issued, the store buffer is
checked for any outstanding stores to the same offset. %
If any exist, the load is, conservatively, stalled until the stores have been committed. %
However, this can cause a timing delay in case of a matching offset,
even if both memory accesses go to different addresses. %

\subsection{UPEC-DIT with Inductive Proofs}

In this subsection, we extend our experiments to the \ac{boom}~\cite{2020-ZhaoKorpan.etal}, a superscalar \RISCV{} processor with \ac{fp} support, a deep 10-stage pipeline and out-of-order execution. % 

In a first attempt, the same unrolled approach (cf.~Sec.~\ref{subsec:unrolled-proofs}), as above for in-order processors, was applied to \ac{boom}, 
and UPEC-DIT was able to prove the data-obliviousness of basic arithmetic instructions and multiplication. %
However, the high design complexity caused by the deep pipeline pushed the formal tools to their limits, with some proof times exceeding 20 hours. %
Fortunately, with a few design-specific adjustments, it was still possible to formally verify the absence of any data-dependent side-effects. %
These optimizations included splitting the verification into individual proofs for the integer and \ac{fp} pipelines, since these are implemented as separate modules in \ac{boom}. %
To do this, we constrained the outputs of the \ac{fp} register file to be equal for the integer pipeline verification and vice versa. %
After unrolling for~7 clock cycles, no new propagation was detected in either pipeline. %
Furthermore, a sound black-boxing (cf.~Sec.~\ref{subsec:black-boxing}) of complex components, such as the \ac{rob} and caches, was performed, significantly reducing the complexity of the computational model: %
While a single instance of \ac{boom} has more than 500k state bits, our final 2-safety model contained only about 47k state bits. %

Black-boxing the \ac{rob} not only reduces the size of the computational model, 
but also drastically reduces the complexity of spurious counterexamples caused by the symbolic initial state. %
During normal operation of an out-of-order processor, the \ac{rob} reflects the control state of the entire system. %
Assuming an arbitrary initial state can therefore lead to many false counterexamples %
where the \ac{rob} in no way reflects the state of the computational pipeline. %
Fortunately however, the actual state of the \ac{rob} is of no concern to UPEC-DIT, as long as instructions commit to it equally between the two instances. % 
We can therefore create a black-box and consider its inputs in our proof statement. %
 
Nevertheless, the long proof times caused by the high complexity made us rethink the approach. %
Thus, we extended the methodology and developed the inductive proof, as presented in Sec.~\ref{sec:methodology}, which spans over only a single clock cycle. %
Eliminating the need to unroll the circuit drastically reduced the complexity of the computational model, diminishing individual proof times from over 20 hours to just a few seconds. %
In this transition, we also found that in most cases, data can only affect the control flow through a few specific channels. %
This insight made it easier to find meaningful invariants and constraints that restrict the over-approximated state space in order to avoid false counterexamples. %

\begin{table}
  \small
  \begin{center}
    \caption{Results for inductive proofs}
    \label{tab:inductive-proof-results}
    \vspace{-1ex}
    \begin{tabular}{lrrrrr}
      \hline &&&&& \\ [-2ex]
      \hline &&&&& \\ [-2ex]
      \textbf{Design} & \textbf{\#State Bits} & $|Z_D|$ & \textbf{\#BB} & \textbf{Time} & \textbf{Mem.}\\
      \hline &&&&& \\ [-2ex]
      BOOM (Int) & 46958 &  31 & 4 &  5s & 1.6GB \\
      BOOM (FP)  & 46958 & 143 & 4 &  5s & 2.0GB \\
      CVA6       & 23484 &  22 & 7 & 58s & 3.1GB \\
      \hline &&&&& \\ [-2ex]
      \hline &&&&& \\ [-2ex]
    \end{tabular}
  \end{center}
  \begin{minipage}{1.0\linewidth}
    For each experiment, the table shows the size of the final computational model, the number of (unconstrained) data signals within per instance, the number of black-boxes per instance, and the time and memory consumption of the final inductive step proof. %
  \end{minipage}
  \vspace{-3ex}
\end{table}

Tab.~\ref{tab:inductive-proof-results} shows the results in terms of
computational complexity for the inductive proofs. %
We also re-run our experiment on CVA6~\cite{2019-ZarubaBenini} to
illustrate the improvement in scalability compared to an unrolled
proof. %
As shown in the table, proof time and memory usage are reduced significantly by eliminating the need to unroll the model. %
In CVA6, data propagates through 22 internal signals~($Z_D$), which can take arbitrary values in the final proof. %
By assuming the equality of two internal signals, we
  excluded division and \ac{csr} operations from consideration. %
Branches and Load/Store instructions were excluded by black-boxing the branch unit and the \ac{lsu}. %

The experiments in \ac{boom} were split for the integer and \ac{fp} pipelines. %
Within the integer pipeline, UPEC-DIT detected data propagation
into~31 internal signals~($Z_D$). %
We excluded branches, division, and misaligned memory accesses after receiving counterexamples for each of these cases. %
Inside the \ac{fp} pipeline, data is propagated into~143 internal data
signals~($Z_D$). %
However, only counterexamples caused by \ac{fp} division and square root had to be excluded from the proof with constraints. %
 
For \ac{boom}, our results show that branch, integer division, \ac{fp} division, and square-root operations have data-dependent timing. %
We have not examined memory access instructions, as these are usually given special treatment in constant-time programming. %
The remaining integer arithmetic instructions (including multiplication) were formally proven to execute data-obliviously. %
An interesting case are \ac{fp} arithmetic instructions: %
Although their timing is independent of their operands, they are not data-oblivious. %
UPEC-DIT revealed that these instructions can still leave a side effect in the form of an exception flag in the \ac{fp} \ac{csr} named \textit{fcsr}. %
This is in compliance with the \ac{isa} specification of \RISCV{}. %
However, in order to create a timing side channel, a victim program would have to explicitly load these registers. %
Therefore, these instructions could, in fact, be used securely for constant-time programming, if a suitable \ac{sw} restriction was introduced. %

\section{Conclusion}
\label{sec:conclusion}

In this paper, we proposed UPEC-DIT, a novel methodology for formally verifying data-independent execution in RTL designs. %
We proposed an approach based on an inductive property over a single clock cycle, which facilitates a verification methodology scalable even to complex out-of-order processors.
We presented and discussed several techniques that can help the verification engineer simplify and accelerate the verification process. %
UPEC-DIT was evaluated against several open-source designs, ranging from small functional units to a complex out-of-order processor. %
While many of the implemented instructions execute as expected, UPEC-DIT uncovered some unexpected timing violations. %
Our future work will address the design of security-conscious hardware. %
We envision that UPEC-DIT, if integrated in standard design flows, can make significant contributions to restoring the trust in the hardware for confidential computing. %
\bibliographystyle{IEEEtran}
% Generated by IEEEtran.bst, version: 1.14 (2015/08/26)

\vspace{-7ex}

\begin{IEEEbiography}[{\includegraphics[width=1in,height=1.25in,clip,keepaspectratio]{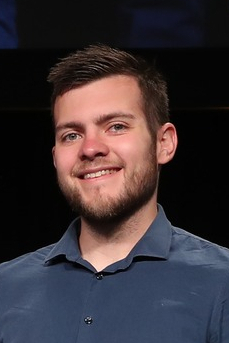}}]{Lucas Deutschmann}
	received the M.Sc.\ %
	degree in Electrical and Computer Engineering from RPTU Kaiserslautern-Landau, Kaiserslautern, Germany, in 2019. %
	He is currently pursuing a doctoral degree with Electronic Design
	Automation group at the same institution. %
	His current research interests include formal verification, hardware
	security, side channel attacks and data-oblivious computing. %
	\vspace{-7ex}
\end{IEEEbiography}

\begin{IEEEbiography}[{\includegraphics[width=1in,height=1.25in,clip,keepaspectratio]{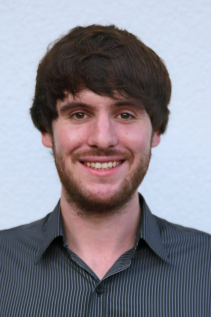}}]{Johannes
		M\"uller}
	received the Dipl.-Ing.\ %
	degree in Electrical and Computer Engineering from RPTU Kaiserslautern-Landau, Kaiserslautern, Germany, in 2018. %
	He is currently pursuing a doctoral degree with Electronic Design
	Automation group at the same institution. %
	His current research interests include formal verification, access
	control in SoCs and microarchitectural timing side channels. %
	\vspace{-7ex}
\end{IEEEbiography}

\begin{IEEEbiography}[{\includegraphics[width=1in,height=1.25in,clip,keepaspectratio]{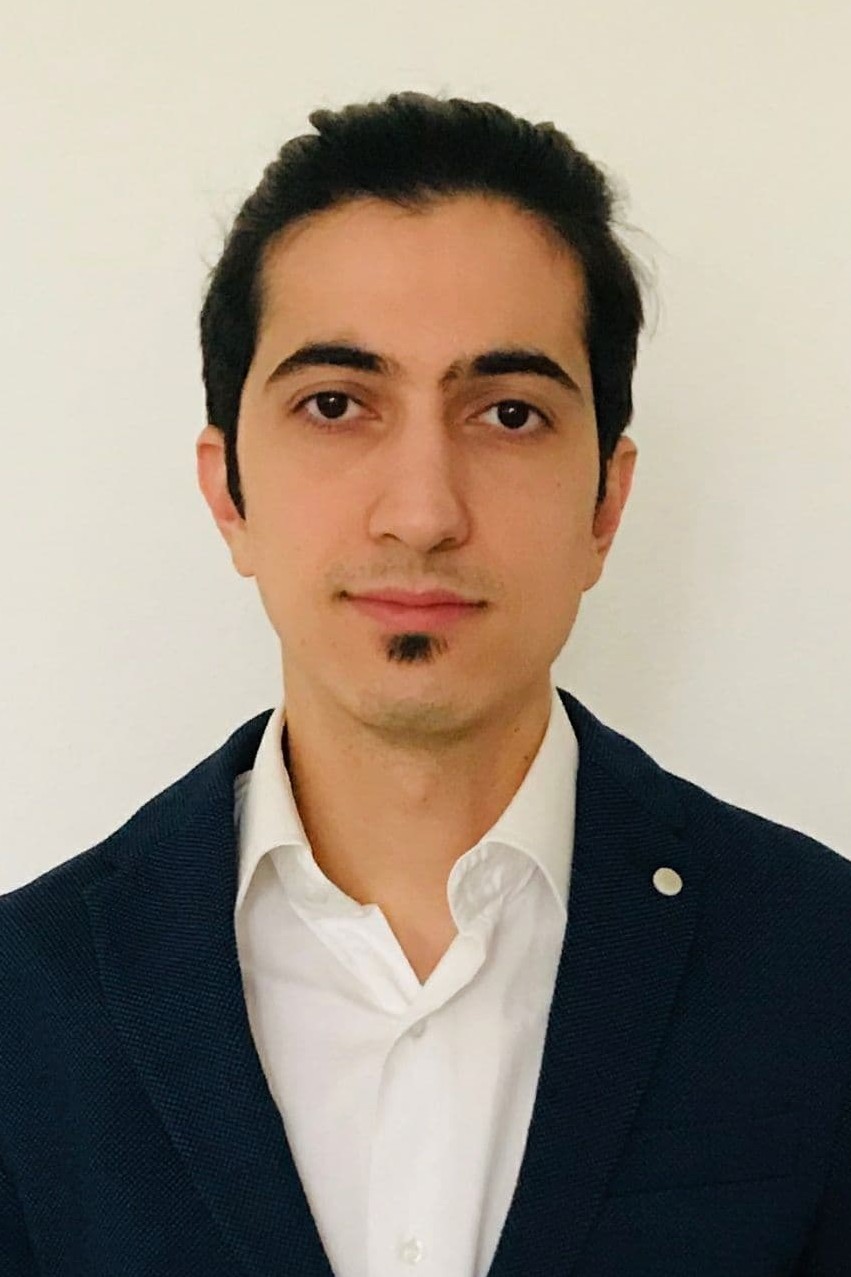}}]{Mohammad Rahmani Fadiheh}
	received the M.Sc.\ %
	degree in Electrical and Computer Engineering from RPTU Kaiserslautern-Landau, Kaiserslautern, Germany, in 2017, and the Dr.-Ing. degree in Electrical Engineering from the same institution in 2022. %
	He is currently working as a postdoctoral researcher in the Department of Electrical Engineering at Stanford University, Stanford, USA. %
	His current research interests include formal verification, hardware
	security, side channel attacks and secure hardware design. %
	\vspace{-7ex}
\end{IEEEbiography}

\begin{IEEEbiography}[{\includegraphics[width=1in,height=1.25in,clip,keepaspectratio]{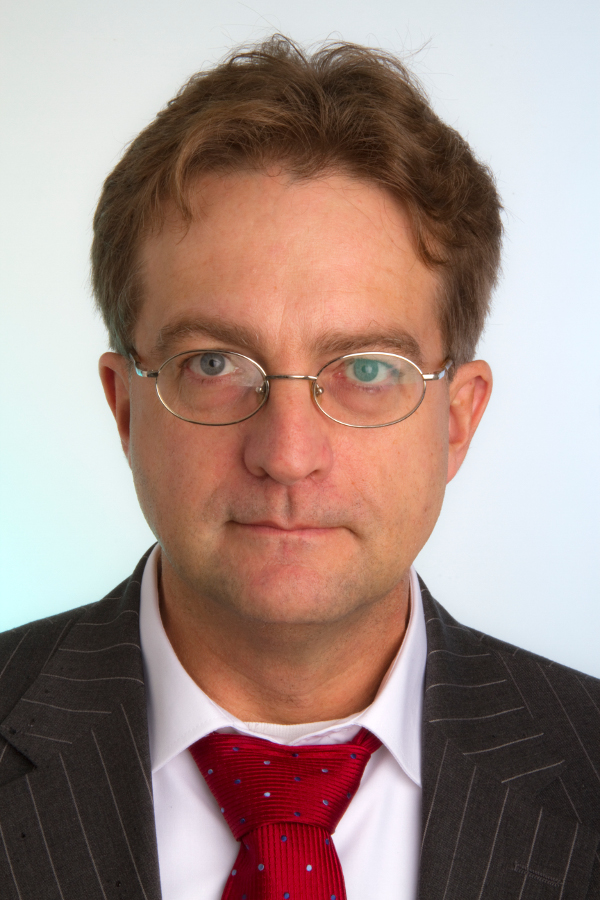}}]{Dominik Stoffel}
	received the Dipl.-Ing.\ %
	degree from the University of Karlsruhe, Karlsruhe, Germany, in
	1992, and the Dr.~phil.~nat degree from Johann-Wolfgang Goethe
	Universit\"at, Frankfurt, Germany, in 1999. %
	He has held positions with Mercedes-Benz, Stuttgart, Germany, and
	the Max-Planck Fault-Tolerant Computing Group, Potsdam, Germany. %
	Since 2001, he has been a Research Scientist and a Professor with
	the Electronic Design Automation group, RPTU Kaiserslautern-Landau, Kaiserslautern, Germany. %
	His current research interests include design and verification
	methodologies for Systems-on-Chip. %
	\vspace{-7ex}
\end{IEEEbiography}

\begin{IEEEbiography}[{\includegraphics[width=1in,height=1.25in,clip,keepaspectratio]{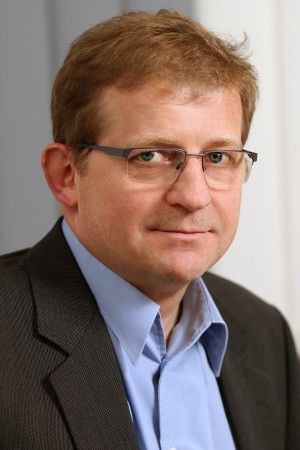}}]{Wolfgang Kunz}
	received the Dipl.-Ing.\ %
	degree in Electrical Engineering from University of Karlsruhe in
	1989, the Dr.-Ing. %
	degree in Electrical Engineering from University of Hannover in 1992
	and the Habilitation degree from University of Potsdam (Max Planck
	Society, Fault-Tolerant Computing Group) in 1996. %
	Since 2001 he is a professor at the Department of Electrical and
	Computer Engineering at RPTU Kaiserslautern-Landau, Kaiserslautern, Germany. %
	
	Wolfgang Kunz conducts research in the area of System-on-Chip design
	and verification. %
	His current research interests include verification-driven design
	methodologies for hardware and firmware, safety analysis and
	hardware security. %
	\vspace{-7ex}
\end{IEEEbiography}

\end{document}